\documentclass[9pt,shortpaper,twoside,web]{ieeecolor}
\usepackage{generic}
\usepackage{cite}
\usepackage{amsmath,amssymb,amsfonts}
\usepackage{algorithmic}
\usepackage{graphicx}
\usepackage{textcomp}
\usepackage{psfrag} 
\usepackage{cuted}
\usepackage{float}
\usepackage{graphicx} 
\usepackage[fancythm,fancybb,morse,ieee]{jphmacros2e} 
\usepackage{array}
\usepackage{bbm}
\usepackage{mathrsfs}
\usepackage{epstopdf}
\DeclareMathAlphabet{\mathpzc}{OT1}{pzc}{m}{it}
\usepackage{graphicx}
\usepackage{amsmath}
\usepackage[fancythm,fancybb,morse,ieee]{jphmacros2e}
\usepackage{caption}
\usepackage{subcaption}

\def\BibTeX{{\rm B\kern-.05em{\sc i\kern-.025em b}\kern-.08em
    T\kern-.1667em\lower.7ex\hbox{E}\kern-.125emX}}
\begin{document}
\title{Nonlinear Discrete-time System Identification without Persistence of Excitation: Finite-time Concurrent Learning Methods}
\author{Farzaneh Tatari, Christos Panayiotou, and Marios Polycarpou 
\thanks{This work is funded by the European Union's Horizon 2020 research and innovation programme under grant agreement No. 739551 (KIOS CoE).}
\thanks{The authors are with the KIOS Research and Innovation Center of Excellence and the Department of Electrical and Computer Engineering, University of Cyprus, Nicosia,	Cyprus.     (e-mail: tatari.farzaneh@ ucy.ac.cy, christosp@ucy.ac.cy, mpolycar@ucy.ac.cy). }}

\maketitle
\begin{abstract}
This paper deals with the problem of finite-time learning for unknown discrete-time nonlinear systems' dynamics, without the requirement of the persistence of excitation. Two finite-time concurrent learning methods are presented to approximate the uncertainties of the discrete-time nonlinear systems in an online fashion by employing current data along with recorded experienced data satisfying an easy-to-check rank condition on the richness of the recorded data which is less restrictive in comparison with persistence of excitation condition. For the proposed finite-time concurrent learning methods, rigorous proofs guarantee the finite-time convergence of the estimated parameters to their optimal values based on the discrete-time Lyapunov analysis. Compared with the existing work in the literature, simulation results illustrate that the proposed methods can timely and precisely approximate the uncertainties. 
\end{abstract}


\section{Introduction}
\label{sec:introduction}
Learning a high-fidelity model of a nonlinear system via stream of data is of vital importance in many engineering applications since such systems are highly subjected to uncertainties that can degrade the performance of the system controllers. It is well known that many learning strategies, such as least-square and gradient descent [1], depend heavily on the persistency of excitation (PE) condition that permanently requires a complete span of the space over which the learning is performed, and failure to fulfill this condition will lead to poor learning results. However, the PE condition might be hard to achieve or even might not be feasible in some scenarios, especially in the context of online learning. 

Concurrent learning [2]-[6] has emerged as a promising paradigm in the direction that guarantees the exponential convergence of the approximated parameters to their optimal values with relaxing the strict assumption of the PE condition to some easy-to-check verifiable conditions on the richness of data. Concurrent learning technique benefits from recorded experienced data along with current data to replace the PE condition on the regressor with a rank condition on the memory stack of the regressor recorded data. Based on this rank condition, the regressor matrix of the recorded data must contain the same number of linearly independent elements as the dimension of the independent basis functions in the regressor. 

In many practical situations, the system dynamics are employed in online monitoring and control applications; therefore, learning the system’s unknown dynamics over a finite-time interval is required. Finite-time learning is of more interest rather than learning with asymptotic or exponential convergence rate since it is more physically realizable than concerning infinite time. Moreover, such finite-time learning scheme is of utmost importance for learning-based controlled systems that demand fast and reliable actions. The knowledge of a finite time for parameter estimation error convergence in systems' control improves the performance while avoiding conservationism due to slow or asymptotic convergence. 

Although, finite-time control methods have been extensively employed for discrete and continuous-time systems [7]-[11], fewer attempts have been proposed by several researchers to tackle finite-time learning such as [12]-[17] where the majority of them are concerning with continuous-time systems. 
Using concurrent learning, the authors in [15] proposed a finite-time learning; however, the results and analysis are limited to continuous-time systems.

It is favorable to employ finite-time learning schemes that can alleviate the restrictive PE condition for discrete-time systems' identification. Discrete-time systems are quite different with continuous-time systems; therefore, the tools applied to the continuous-time domain cannot be directly employed for the discrete-time domain.  Moreover, different from finite-time stability analysis of continuous-time systems, which can draw support from many mathematical tools, the mathematical tools for finite-time stability analysis of discrete-time systems are not plenty. Therefore, the research on finite-time convergent concurrent learning identification method for discrete-time systems is more challenging and complex. 
 
 In many practical applications, however, identifying precise discrete-time dynamics of the system is required due to the development of computer technology and the introduction of digital controllers and sensors. Hence, it is of great practical importance to investigate finite-time learning of discrete-time systems. There are only a few results discussing the identification of discrete-time systems including [16]-[19]. The work of [16], studied the discrete-time systems' uncertainty identification in finite-time, where the proposed batch learning method required the online invertibility check of a regressor matrix and its inverse computation, along with interval excitation of the regressor. However, the required regressor matrix inversion makes the method in [16] inefficient in online learning for the case of large number of unknown parameters' identification. The authors in [17] presented some results on parameter estimation via dynamic regressor extension and mixing for both continuous-time and discrete-time systems, however, their results on finite-time convergence are limited to continuous-time scalar systems. Moreover, none of the approaches given in [16] and [17] have computed the upper bound of the settling-time function for convergence. 
 In [18], a concurrent learning-based method for discrete-time function approximation is presented that relaxed the PE condition and guaranteed the asymptotic convergence of the estimated parameters. Our earlier work [19] studied how concurrent learning can be employed for finite-time identification of discrete-time systems' dynamics. However, [19] only investigated the special case of adaptive approximators with zero minimum functional approximation error (MFAE) where MFAE is the residual approximation error in the case of optimal parameters. 
 
 Motivated by the above-mentioned discussions, this work aims to propose online finite-time concurrent learning (FTCL) schemes for discrete-time systems that guarantee finite-time parameter convergence without the restrictive PE condition by employing a memory stack of data, satisfying a rank condition. In this paper, in contrast to [19], every proposed FTCL method includes both cases of identification where the optimal set of unknown parameters can make the identification error either non-zero or zero. For the systems with mismatch identification error, the MFAE is non-zero. Moreover, opposed to [16], the proposed adaptive FTCL methods do not require the regressor matrix inversion and they represent finite upper bounds for the settling-time functions. In order to approximate the unknown discrete-time system functions, linearly parameterized universal approximators such as radial basis function neural networks are used. It is shown that under a verifiable rank condition and along with a learning rate condition, the proposed FTCL methods guarantee the finite-time convergence of the parameters' estimation errors. 

This paper
contains the following contributions:
 	\begin {enumerate}
	\item Two novel FTCL schemes are presented for learning the unknown dynamics of nonlinear discrete-time systems.
    \item For the two FTCL methods, rigorous proofs ensure the finite-time convergence of the parameters estimation errors to the origin for adaptive approximators with zero MFAE using discrete-time Lyapunov analysis. It is also guaranteed that for adaptive approximators with non-zero MFAE, the parameter estimation errors are finite-time attractive.
    \item The finite upper bounds for the settling-time functions of the proposed FTCL methods are given. In addition, based on the finite-time analysis, for both presented FTCL methods, conditions on the learning rates are derived for finite-time convergence.	
	\end {enumerate}
	\paragraph*{Notation}
  $\Real$, $\mathbb{Z}$, and $\mathbb{N}^+$ respectively show the set of real, integer and natural numbers without zero. $\Arrowvert . \Arrowvert$ denotes the Euclidean norm for vectors and induced 2-norm for matrices.  Trace of a matrix is indicated with $tr(.)$. The minimum and maximum eigenvalues of matrix $A$ are respectively denoted by $\lambda_{min} (A)$ and $\lambda_{max}(A)$. The matrix $I$ is the identity matrix of appropriate dimensions. $\lfloor . \rfloor: \Real \mapsto \mathbb{Z} $  is the floor function. 

\section{Problem Formulation and Preliminaries}
	\subsection{Preliminaries}
	
	The following definitions, facts and lemmas are needed through the paper. 
			\begin{definition}{[20] }
		The bounded signal $d(k)$ is said to be persistently exciting if there exist positive scalars $\mu_{1}$, $\mu_{2}$ and $T\in \mathbb{N}^+$ such that $ \forall \tau \in \mathbb{N}^+$,
		$\mu_{1} I \le \sum_{k=\tau}^{\tau+T} d(k) d^T(k) \le \mu_{2}I $.
	\end{definition}
\begin{definition} {[21]}
	Consider the system 
		\begin{align}
		y(k+1)= F(y(k)),\,\,\, y(0)=y_0, \label{sys2}
		\end{align} where $y\in \mathcal{D}_y \subset \mathbb{R}^n$, $F:\mathcal{D}_y \mapsto \Real^n$ is a nonlinear function on the neighborhood $\mathcal{D}_y$ of the origin and origin is the equilibrium point of \eqref{sys2}. The system \eqref{sys2} is said to be  
		\begin {enumerate}
		\item finite-time stable, if it is Lyapunov stable and finite-time convergent where $\forall y_0 \in \mathcal{D}_y$ any solution $y(k)$ of \eqref{sys2} reaches the origin at some finite time moment, i.e., $y(k)=0$, $\forall k > K(y_0)$ where $K:\mathcal{D}_y \backslash \{0\} \mapsto \mathbb{N}^+ $ is a settling-time function.
		\item finite-time attractive to an ultimate bounded set $Y$ around zero, if solution $y(k)$ of \eqref{sys2} reaches $Y$ in finite-time $k>K(y_0)$ and remains there  $\forall k> K(y_0)$, where $K:{\mathcal {D}}_y\backslash \{0\}  \mapsto \mathbb{N}^+$ is a settling-time function.
		\end{enumerate}
	\end{definition}
		\begin{fact}
		For every matrix $A$ and $B$ of the same dimensions, it is known that $\|A\|-\|B\| \le \|A-B\| \le \|A\|+\|B\|$.	
	\end{fact}
	
	\begin{fact}
	For a vector $x= [x_1,x_2,...,x_n]^T\in \Real^n$, the $p$-norm is defined as $\|x\|_p = (\sum_{i=1}^{n}|x_i|^p)^{\frac{1}{p}}$ and for positive constants $r$ and $s$, if $0<r<s$, using H\"{o}lder inequality [22], one has $\|x\|_s \le \|x\|_r \le n^{\frac{1}{r} - \frac{1}{s}} \|x\|_s$.
		\end{fact}

	\begin{lemma} {[19]}
		Consider the system \eqref{sys2}. Suppose there is a continuous positive definite Lyapunov function $V:\Omega \mapsto \Real$ where $\Omega$ is an open neighborhood of the origin and $\Delta V(y(k))= V(y(k))-V(y(k-1))$. If there exist positive constants $0 < a <1$, $b>0$, $0<\mu<1 $,  and a neighborhood $\mathcal{M}\subset \Omega$ such that for $\forall y \in \mathcal{M} \backslash\{0\}$,
		\begin{align}
		\Delta V(y(k))\le-aV(y(k-1)) -b (V(y(k-1)))^ \mu,  \label{DelVlem1}
		\end{align}
		 then, the system \eqref{sys2} is finite-time stable and there is a neighborhood $\mathcal{N}$ and a settling-time function $K(y_0)$ such that for $V(y_0)>{(\frac{b}{1-a})}^{\frac{1}{1-\mu}}$,
		\begin{align}
		K(y_0) \le \lfloor \frac{V(y_0)}{a(\frac{b}{1-a})^{\frac{1}{1-\mu}}+ b{ (\frac{b}{1-a})^{\frac{\mu}{1-\mu}}}} \rfloor +1,\,\,\,y_0 \in \mathcal{N},
		\end{align}
		and $K(y_0)=1$ for $V(y_0) \le {(\frac{b}{1-a})}^{\frac{1}{1-\mu}}$, $y_0 \in \mathcal{N}$.
	\end{lemma}
	
	\begin{lemma} {[23]}
	Consider the system \eqref{sys2} and assume that there exist a continuous positive definite function $V:\Omega \mapsto \Real$ where $\Omega$ is an open neighborhood of the origin and $\Delta V(y(k))= V(y(k))-V(y(k-1))$. If real numbers $\alpha \in (0,1)$, $c>0$ and a a neighborhood $\mathcal{M}\subset \Omega$ exist such that $\forall y \in \mathcal{M} \backslash\{0\}$,
\begin{align}
\Delta V(y(k)) \le -c \min \{\frac{V(y(k-1))}{c}, V^{\alpha}(y(k-1))\},
\end{align} 
	    then the system \eqref{sys2} is finite-time stable and there exist a settling-time function $K(y_0)$ and a neighborhood $\mathcal{N}$ such that for $V(y_0)> c^{\frac{1}{1-\alpha}}$,
	    \begin{align}
	   K(y_0) \le \lfloor log_{[1-cV(y_0)^{\alpha-1}]} \frac{c^{\frac{1}{1-\alpha}}}{V(y_0)} \rfloor + 1,\,\,\, y_0 \in \mathcal{N},
	    \end{align}
	    and $K(y_0)=1$ for $V(y_0)\le c^{\frac{1}{1-\alpha}}$, $y_0 \in \mathcal{N}$.
 	\end{lemma}

	\subsection{Problem Formulation}
	Consider a discrete-time nonlinear system as follows,
	\begin{align} 
	{x}(k+1) =& {f}(x(k))+{g}(x(k)){u}(k) ,\,\,\,\,\, x(0)=x_0,\label{sys}
	\end{align} where ${x} \in \mathcal{D}_x \subset {\Real}^n$ is the measurable state vector and ${u} \in \mathcal{D}_u \subset {\Real}^m$ is the control input vector, $\mathcal{D}_x$ and $\mathcal{D}_u$ are compact sets; ${f}: \mathcal{D}_x \mapsto {\Real}^n$, and ${g}: \mathcal{D}_x \mapsto {\Real}^{n \times m}$ are respectively the unknown nonlinear drift and input terms. 
	This paper aims to learn the system unknown dynamics in \eqref{sys}, namely to approximate the uncertain functions $f(x)$ and $g(x)$ in a finite time using concurrent learning techniques. 
	
	In order to learn $f(x)$ and $g(x)$ in the system dynamics, using linearly parameterized adaptive approximators, one has
	\begin{align} 
	{{f}}(x)= \Theta_f^{*T} {{\varphi}} (x) + e_{f}(x), \,\,\,\,
	{g}(x) = { \Theta_g^{*T}}  {\chi} (x) + e_{g}(x),\label{truefg}
	\end{align}
	where the matrices ${\Theta_f^*} \in \mathcal{D}_{f} \subset {\Real}^{ p \times n}$ and ${\Theta}_g^* \in \mathcal{D}_{g} \subset {\Real}^{q \times n}$ denote the unknown optimal parameters of the adaptive approximation models and ${{\varphi}}: \mathcal{D}_x\mapsto {\Real}^{p}$ and $ {\chi}: \mathcal{D}_x\mapsto {\Real}^{q}$ are the vectors denoting the basis functions, whereas $ p$ and $q$ are, respectively, the number of linearly independent basis functions to approximate ${f}(x)$ and ${g}(x)$. The quantities $e_{f}(x)\in \Real^n$ and $e_{g}(x)\in \Real^{n\times m}$ are, respectively, the MFAEs for ${f}(x)$ and ${g}(x)$, denoting the residual approximation error in the case of optimal parameters. If the unknown functions ${f}(x)$ and ${g}(x)$ are approximated exactly by the adaptive approximators $\Theta_f^{T} {{\varphi}} (x)$ and ${ \Theta_g^{T}}  {\chi} (x)$, respectively, MFAE is zero, i.e., $e_f(x)=e_g(x)=0$.
	
	Using \eqref{truefg}, \eqref{sys} is rewritten as follows
	\begin{align} 
	{{x}(k+1)}  = \Theta^{*T}  z(x(k), u(k))+ \varepsilon (x(k), u(k)),\label{regsys}
	\end{align} 
	where $\Theta^*=[\Theta_f^{*T},\Theta_g^{*T}]^T \in \Real ^{(p+q) \times n}$, $z(x,u)=[ \varphi^T(x),u^T  \chi^T(x)]^T\in \Real ^{(p+q)}$, and $\varepsilon(x,u)= e_{f}(x) + e_{g}(x)u$. 

\begin{assumption}
		In the compact set $\mathcal{D}_x$, the approximators' basis functions are bounded, and the approximation error $\varepsilon(x(k),u(k))$ is upper bounded for admissible controls $u(k)$ by a bound $b_\varepsilon \ge 0$ (i.e., $\sup \limits_{x\in \mathcal{D}_x, u\in \mathcal{D}_u} \|\varepsilon(x(k),u(k))\| \le b_{\varepsilon}$).
	\end{assumption}
\begin{remark}
		In the literature, Assumption 1 is standard based on universal approximator characteristics [20].
\end{remark}

Since $x(k+1)$ is not available, regressor filtering [1], [3], [19] is used which  gives the state space solution of \eqref{regsys} as follows 
	\begin{align} 
	&x(k)=  \Theta^{*T} d(k) - l(k) + {C^k} x_0 + \varepsilon_{f}(k), \label{filtx} \\
	&d(k+1) = c d(k) + z(x(k), u(k)), \,\, d(0)=0, \label{regdl}\nonumber \\
	&l(k+1)=Cl(k) + Cx(k), \,\, l(0)=0,
	\end{align}
	where $C=cI, -1<c<1$, $l(k) = \sum_{h=0}^{k-1} C^{k-h} x(h) $ is the filtered regressor of $ x(k)$, $d(k) = \sum_{h=0}^{k-1}  c^{k-h-1} z(x(h), u(h)) $ is the filtered regressor of $z(x(k), u(k))$,
	$\varepsilon _{f}(k) = \sum_{h=0}^{k-1}  C^{k-h-1} \varepsilon (x(h), u(h))$.

Dividing \eqref{filtx} to the signal $n_s := 1 + d^T(k)d(k)  + l^T(k)l(k) $, normalizes \eqref{filtx} as follows
 \begin{align} 
 \bar x(k)= & \Theta^{*T} \bar d(k) - \bar l(k)  + C^k \bar x_0 + \bar \varepsilon(k),   
 \end{align}
where $\bar d = \frac{d}{n_s}$, $\bar l= \frac{l}{n_s}$, $\bar x= \frac{x}{n_s}$ and $ \bar \varepsilon= \frac{\varepsilon_{f}}{n_s}$. Based on Assumption 1, $\bar \varepsilon(k)$ is also upper bounded by a bound $b_{\bar \varepsilon} \ge 0$, i.e., $\|\bar \varepsilon(k)\| \le b_{\bar \varepsilon}$, and  $\|\bar d(k)\|<1$. 
	
	Now, let the approximator be of the form 
	\begin{align} 
	\hat{\bar x}(k)=& \hat\Theta^T(k) \bar d(k) - \bar l(k) + C^k\bar x_0, \label{app}
	\end{align} 
	where $\hat\Theta(k)=[{\hat{\Theta}}_f^T(k),{\hat{ \Theta}}_g^T(k)]^T \in \Real ^{( p+q)\times n}$, ${\hat{\Theta}}_f(k)$ and ${\hat{\Theta}}_g(k)$ are the estimation for parameters matrices  $\Theta^*$, $\Theta_f^*$ and $\Theta_g^*$ at time $k$, respectively. Define the state estimation error as
	\begin{align} 
	e(k) = \hat{ \bar { x}}(k) - \bar {x}(k) = \tilde \Theta^T(k) \bar d(k) - \bar \varepsilon(k),\label{errork}
	\end{align}
	where $\tilde \Theta (k) := \hat \Theta (k) - \Theta^*  := [{\tilde\Theta}_f^T(k),{\tilde\Theta}_g^T(k)]^T$ is the parameter estimation error with ${\tilde\Theta}_f(k) := {\hat\Theta}_f(k) - \Theta_f^*$, ${\tilde\Theta}_g(k) := {\hat\Theta}_g(k) -\Theta_g^*$.
	
	To fulfill the finite-time learning of the uncertainties $f(x)$ and $g(x)$ in the system \eqref{sys}, the paper objective is to propose finite-time concurrent learning-based estimation methods that relax the PE condition and satisfy the following criteria:
		\begin {enumerate}
	\item For adaptive approximators with zero MFAE, the parameters' estimation error $\tilde \Theta (k) $ converges to zero in finite time.
		\item For adaptive approximators with non-zero MFAE, the parameters' estimation error $\tilde \Theta (k) $ is finite-time attractive to a bounded set around zero.
		\end {enumerate}

\section{Finite-time Concurrent Learning for Unknown Discrete-time Systems}
	
	 In order to use concurrent learning, that employs recorded experienced data concurrently with current data in the parameter estimation law, the past data is recorded and stored in the memory stacks $M \in \Real^{(p+q) \times P}$, $L \in \Real^{n \times P}$ and $X \in \Real^{n \times P}$, at time steps $\tau_1,...,\tau_{P}$ as given below,
\begin{align} 
	&M=[\bar d(\tau_1),\bar d(\tau_2),..., \bar d(\tau_{P})],\,\,\,L=[\bar l(\tau_1),\bar l(\tau_2),..., \bar l(\tau_{P})],\nonumber \\& X=[\bar x(\tau_1),\bar x(\tau_2),..., \bar x(\tau_{P})], \label{stacks}
	\end{align}
 where $P$ denotes the number of data points stored in every stack. Note that $P$ is chosen such that $M$ contains at least as many linearly independent elements as the dimension of $d(k)$ (i.e., the total number of linearly independent basis functions for approximating $f(x)$ and $g(x)$), given in \eqref{filtx}, that is called as $M$ rank condition and requires $P \ge p+q$. Consider the error $e_h(k)$ for the $h^{th}$ recorded data as
\begin{align} 
	e_h(k) = \hat {\bar {x}}_h(k) - \bar { x}(\tau_h) , \label{errorkh}
	\end{align}
 where 
	\begin{align} 
	\hat {\bar { x}}_h(k) =& \hat\Theta^T(k) \bar {d}(\tau_h)- \bar l(\tau_h) + C^{k}  \bar x_0, \label{hatxkh}
	\end{align}
	is the state estimation at $0\le \tau_h < k$ time step, $h=1,...,P$, employing the current estimated parameters matrix $\hat\Theta(k)$ and recorded $\bar {d}(\tau_h)$ and $\bar l(\tau_h)$. Replacing $\bar x(\tau_h)$ into \eqref{errorkh}, one has 
	\begin{align} 
	e_h(k) = \tilde \Theta^T(k) \bar d (\tau_h) - \bar \varepsilon(\tau_h). \label{errorkhA}
	\end{align} 
 \begin{remark}
	It is possible to meet the rank condition by selecting and recording data either during a normal course of operation or when the system is excited over a finite time interval [2].
	\end{remark}
	In the following, two FTCL methods are presented to estimate the parameters $\hat\Theta(k)$ for the system approximator \eqref{app} in finite time. 
	
	\emph{FTCL Method 1:}
	The first proposed FTCL law for estimating the parameters of the system approximator is as follows
	\begin{align}
	{\hat \Theta}(k+1) = & {\hat \Theta}(k) - \Gamma[\Xi_{G}\bar d(k)  e^T(k)  + \Xi_{C} (\sum_{h=1}^{P} \bar d(\tau_h)  e_h^T(k) \nonumber \\& +  \frac{  \sum_{h=1}^{P} \bar d(\tau_h)  e_h^T(k)}{\beta+\|\sum_{h=1}^{P} \bar d(\tau_h)  e_h^T(k)\|})],  \label{Estrule1}
	\end{align}
	where 
	$\Gamma=\gamma I$ is the learning rate matrix with positive constant $\gamma>0$, $\beta$ is a design constant parameter satisfying $\beta >  P b_{\bar \varepsilon}$, $\Xi_{C}= \xi_{C} I$ and $\Xi_{G}= \xi_{G} I$ with positive constants $ \xi_{C}>0$ and $\xi_{G}>0$. 
	
	\emph{FTCL Method 2:}
	The second proposed FTCL law for finite-time parameters' estimation for approximator \eqref{app} is 
	\begin{align}
	{\hat \Theta}(k+1) = & {\hat \Theta}(k) - \bar\Gamma[\bar\Xi_{G}\bar d(k) \lfloor e^T(k) \rceil ^ {\gamma_1}  + \bar\Xi_{C} \sum_{h=1}^{P} \bar d(\tau_h) \lfloor e_h^T(k)\rceil ^ {\gamma_1} ],  \label{Estrule2}
	\end{align}
	where $\lfloor . \rceil ^ {\gamma_1}:= |.|^ {\gamma_1} sign(.)$ such that $|.|$ and $sign(.)$ are component-wise operators and $0 < {\gamma_1}<1$, $\bar\Gamma=\bar\gamma I$ is the learning rate matrix with constant $\bar\gamma>0$, $\bar\Xi_{C}= \bar\xi_{C} I$ and $\bar\Xi_{G}= \bar\xi_{G} I$ with constants $ \bar\xi_{C}>0$ and $\bar\xi_{G}>0$. 
	
	The above estimation laws \eqref{Estrule1} and \eqref{Estrule2} have two learning terms where the term $\Xi_{G}\bar d(k)  e^T(k)$ in \eqref{Estrule1} ($\bar\Xi_{G}\bar d(k)  \lfloor e^T(k) \rceil ^ {\gamma_1} $ in \eqref{Estrule2}), containing the current state approximation error, is widely used in the gradient descent method and the term $\Xi_{C}  (\sum_{h=1}^{P} \bar d(\tau_h)  e_h^T(k) +  \frac{  \sum_{h=1}^{P} \bar d(\tau_h)  e_h^T(k)}{\beta+\|\sum_{h=1}^{P} \bar d(\tau_h)  e_h^T(k)\|})$ in \eqref{Estrule1} ($\bar\Xi_{C} \sum_{h=1}^{P} \bar d(\tau_h) \lfloor e_h^T(k)\rceil ^ {\gamma_1}  $ in \eqref{Estrule2}), containing the past experienced data, is called the concurrent learning term. For the learning weights $\Xi_{C}$ and $\Xi_{G}$ in \eqref{Estrule1} ($\bar\Xi_{C}$ and $\bar\Xi_{G}$ in \eqref{Estrule2}), the constants $\xi_{C}$ and $\xi_{G}$ ($\bar\xi_{C}$ and $\bar\xi_{G}$) are, respectively, set such that one of the two learning terms can be prioritized over the other. 
	
	
	 \begin{remark}  Before the completion of the first $P$ steps of learning, required for filling the data stacks in (14), we set $\Xi_C=0$ and $\bar \Xi_C=0$, respectively, in \eqref{Estrule1} and \eqref{Estrule2} such that they only employ current data to update the estimated parameters.\end{remark}
\begin{remark} 
	FTCL method 2, in comparison with FTCL method 1, does not need the knowledge of $b_{\bar\varepsilon}$ where in \eqref{Estrule1} $\beta$ should satisfy $\beta> P b_{\bar\varepsilon}$.\end{remark}
 \begin{remark} 
 In contrast to the studies that design the controller to extract rich data for system identification [24], this paper objective is to identify the unknown dynamics regardless of the controller design.
\end{remark}
	
\section{Finite-time convergent analysis for the proposed FTCL Methods}	
	In this section, the finite-time convergence properties of the proposed learning methods are given. It should be noted that in the proposed FTCL methods, the stored data in $M$ and other stacks is selected based on data recording algorithm in [25] to maximize $\frac{\lambda_{min}(S)}{\lambda_{max}(S)}$ where $ S=\sum_{h=1}^{P} \bar d(\tau_h)\bar d^T(\tau_h)$, and due to the satisfaction of $M$ rank condition, $ S>0$.
	\begin{theorem}
		Consider the approximator for nonlinear system \eqref{sys} given by \eqref{app}, whose parameters are estimated using the estimation law \eqref{Estrule1} with the regressor given by \eqref{regdl}. Let Assumption 1 hold. Once the $M$ rank condition and 
		\begin{align}
			\gamma < \frac{2 \xi_{C} \lambda_{min}(S) } { (\xi_{G}  + \xi_{C} \lambda_{max}(S) (1+\frac{1}{\beta}))^2}, \label{gamacond}
			\end{align} are satisfied, then
	\begin{enumerate}
		\item for adaptive approximators with zero MFAE, i.e., $\bar \varepsilon(k)=0$, the parameter estimation law \eqref{Estrule1} ensures that $\tilde \Theta(k)$ converges to zero within finite time steps and a settling-time function
		\begin{align}
	K_1^*(\tilde\Theta_0) \le \lfloor \frac{V(\tilde\Theta_0)}{\mathpzc{a_{\gamma}}(\frac{\mathpzc{b_{\gamma}}}{1-\mathpzc{a_{\gamma}}})^2+ \mathpzc{b_{\gamma}}(\frac{\mathpzc{b_{\gamma}}}{1-\mathpzc{a_{\gamma}}})} \rfloor +1; \label{k1}
	\end{align}   
	\item  for adaptive approximators with non-zero MFAE, i.e., $\bar \varepsilon(k) \ne 0$, ($\|\bar \varepsilon(k) \| \le b_{\bar \varepsilon}$), the parameter update law \eqref{Estrule1} guarantees that ${\tilde \Theta}(k)$ is finite-time attractive to the bound,  
	\begin{align}
	\|\tilde \Theta(k)\| \le b_{\tilde \Theta},\,\,\,b_{\tilde \Theta} = \frac{-\mathpzc{b_u} - \sqrt{(\mathpzc{b_u})^2 - 4(\mathpzc{a})\mathpzc{c}}}{2\mathpzc{a}}, \label{btheta} 
	\end{align} with a settling-time function
	\begin{align}
		  K_2^*(\tilde\Theta_0) \le \lfloor  \frac{V(\tilde\Theta_0)-\gamma^{-1} (b_{\tilde\Theta})^2}{\mathpzc{a_{\gamma}} (\gamma^{-1} (b_{\tilde\Theta})^2) -\mathpzc{b_u} \|\tilde \Theta_0\| - \mathpzc{c}}  \rfloor +1, \label{k2}
		 \end{align} 
	such that 
	\begin{align}
	&\mathpzc{a_{\gamma}} = -{\gamma}\mathpzc{a} ,\,\,\, \,\,\, \,\,\, \,\,\, \,\,\, \mathpzc{b_{\gamma}} =    \frac{2 \sqrt {\gamma} \xi_{C} }{(\eta +1)} (\frac{\lambda_{min}(S)}{\lambda_{max}(S)})  ,\label{agama}  
	\\&\mathpzc{a}= -2\xi_{C}\lambda_{min}(S) +\xi_{G}^2 \gamma  + 2 \gamma \xi_{G} \xi_{C} {\lambda_{max}(S) } (\frac{1}{\beta}+1)  \nonumber \\ & + \gamma \xi_{C}^2 \lambda_{max}^2(S) (1+\frac{1}{\beta})^2, \label{a}  \\ 
		&\mathpzc{b_u}= \frac{2 \xi_{C} }{(\eta +1)} (\frac{\lambda_{min}(S)}{\lambda_{max}(S)})+ 2 b_{\bar \varepsilon}(\xi_{C}   \gamma \lambda_{max}(S)(\xi_{G} \nonumber \\ & +  \xi_{C} P(1+ \frac{1}{\beta^2}))  +  \xi_{C} P(\gamma \xi_{G}    + 1) +\xi_{G}  +  \gamma \xi_{G}^2),  \label{bu}  \\ &\mathpzc{c}= b_{\bar \varepsilon}[ \frac{2 \xi_{C}{P} }{\lambda_{min}(S)} + \gamma( \xi_{G}^2 b_{\bar \varepsilon} + 2 \xi_{C}\xi_{G} +  P(  2 \xi_{C}\xi_{G} b_{\bar \varepsilon} + \xi^2_{C} P b_{\bar \varepsilon} \nonumber \\ &  +2 \xi^2_{C}+ 2 \xi^2_{C} \frac{\lambda_{max}(S) }{\lambda_{min}(S)}+ \frac{2 \xi_{G}\xi_{C} }{\lambda_{min}(S)}  +\frac{\xi^2_{C} P b_{\bar \varepsilon}} {\beta^2}  ))]. \label{c} 
	\end{align} 
		\end{enumerate}		
	\end{theorem}

	\begin{proof}
	Please see Appendix.
	
		\begin{remark}
		Contrary to [19] that only considered the special case of zero MFAE for adaptive approximators, Theorem 1 represents rigorous proofs for finite-time convergence of adaptive approximators with non-zero MFAEs.
		\end{remark}
	
	\begin{theorem}
	Let Assumption 1 hold and consider the approximator for system \eqref{sys} given in \eqref{app}, whose parameters are estimated using the estimation law of \eqref{Estrule2} with $0 < {\gamma_1}<1$ and a regressor given in \eqref{regdl}. Once the rank condition on $M$ and 
		\begin{align}
		&	\bar\gamma<\frac{A}{B}, \label{gamacond1}
		\\&	A= 2 \xi_{G} \lambda_{min}^{\frac{\gamma_1+1}{2}}(\bar D(k-1))+2 \xi_{C} \lambda_{min}^{\frac{\gamma_1+1}{2}}(S),\label{A}\\&
		B=\xi_{G}^2 n^{1-\gamma_1} \lambda_{max}^{\frac{\gamma_1+1}{2}}( D(k-1))  + \xi_{C}^2 n^{1-\gamma_1} \lambda_{max}^{\frac{\gamma_1+1}{2}}(S) \nonumber \\& +2 \xi_{C} \xi_{G} n^{1-\gamma_1} \lambda_{max}^{\frac{\gamma_1+1}{2}}(S). \label{B}
		\end{align} 
	are satisfied, then 
	\begin{enumerate}
	    \item 	for adaptive approximators with zero MFAEs, $\bar \varepsilon(k)=0$, the estimation law \eqref{Estrule2} ensures the finite-time convergence of $\tilde \Theta(k)$ to zero and the settling-time function ${\bar K}_1^*(\Theta_0)$ satisfies
	\begin{align}
			&{\bar K}_1^*(\tilde\Theta_0) \le \lfloor \log_{[1-\alpha'V(\tilde\Theta_0)^{\frac{\gamma_1-1}{2}}]} \frac{\alpha'^{\frac{2}{1-\gamma_1}}}{V(\tilde\Theta_0)}\rfloor +1; \label{K111}
		\end{align} 
	\item for adaptive approximators with non-zero MFAE, i.e., $\bar \varepsilon(k) \ne 0$, ($\|\bar \varepsilon(k) \| \le b_{\bar \varepsilon}$), the parameter update law \eqref{Estrule2} guarantees that ${\tilde \Theta}(k)$ is finite-time attractive to the bound $\|	\tilde \Theta(k) \|\le \bar b_{\tilde \Theta}$ with the settling-time function
	\begin{align}
		  {\bar K}_2^*(\tilde\Theta_0) \le \lfloor  \frac{V(\tilde\Theta_0)-\gamma^{-1} (\bar b_{\tilde\Theta})^2}{{\alpha'} \gamma^{-\frac{\gamma_1+1}{2}} (\bar b_{\tilde\Theta})^{\gamma_1+1} -{b'} \|\tilde \Theta_0\| - c'}  \rfloor +1, \label{k22}
		 \end{align} 
such that $\bar b_{\tilde \Theta}$ is the positive root of equation in \eqref{DelV20},	 
		\begin{align}
		&\alpha'=a' \bar\gamma^{\frac{\gamma_1+1}{2}},\,\,\,a'=A -\bar \gamma B, \label{aprime}\\
		&b'=2b_{\bar \varepsilon}^{\gamma_1} n^{\frac{1-\gamma_1}{2}} (\bar \xi_G  + \bar \gamma \bar \xi_G^2 n^{\frac{1-\gamma_1}{2}} + \bar \gamma \bar \xi_C \bar \xi_G n^{\frac{1-\gamma_1}{2}} \nonumber \\&+ \lambda_{max}^{\frac{1}{2}}(S)(\bar \xi_C + \bar \gamma  \bar \xi_C \bar \xi_G n^{\frac{1-\gamma_1}{2}}   + \bar \gamma \bar \xi_C^2 n^{\frac{1-\gamma_1}{2}})),\label{bprime} \\&
		c'=\bar \gamma( n^{1-{\gamma_1}} b_{\bar \varepsilon}^{2\gamma_1}(\bar \xi_{G}^2  +2 \bar \xi_{C} \bar \xi_{G} {P})  +  \bar \xi_{C}^2 {P^2} n^{\frac{1-\gamma_1}{2}} b_{\bar \varepsilon}^{\gamma_1} ).\label{cprime}
		\end{align}
	\end{enumerate}
	\end{theorem}

\end{proof}

\begin{remark}
The settling-time functions satisfying \eqref{k1} and \eqref{K111} are, respectively, obtained for $V(\tilde\Theta_0)>(\frac{\mathpzc{b_{\gamma}}}{1-\mathpzc{a_{\gamma}}})^2$ and $V(\tilde\Theta_0)>\alpha'^{\frac{2}{1-\gamma_1}}$. In Theorems 1 and 2, one has $\bar K_1^*(V(\tilde\Theta_0))=\bar K_2^*(V(\tilde\Theta_0))=1$, respectively, for $V(\tilde\Theta_0)\le(\frac{\mathpzc{b_{\gamma}}}{1-\mathpzc{a_{\gamma}}})^2$ and $V(\tilde\Theta_0)\le\alpha'^{\frac{2}{1-\gamma_1}}$.        
\end{remark}

\begin{remark}
	To have a faster convergence time in \eqref{k1} and \eqref{K111}, maximizing $\mathpzc{b_{\gamma}}$ and $\alpha'$, respectively given in \eqref{agama} and \eqref{aprime}, by maximizing $\frac{ \lambda_{min}(S) }{\lambda_{max}(S)}$, is beneficial. Maximizing $\frac{ \lambda_{min}(S) }{\lambda_{max}(S)}$ (maximizing $ \lambda_{min}(S) $ and minimizing $\lambda_{max}(S)$), also leads to enlarging $\mathpzc{a_{\gamma}}$ and $\alpha'$ which causes faster settling-time in \eqref{k2} and \eqref{k22}.
Therefore, while applying the FTCL methods 1 and 2, the data recording algorithm in [25] is used and the recorded data is selected to maximize $\frac{\lambda_{min}(S)}{\lambda_{max}(S)}$. 
	Moreover, the employed data recording algorithm [25], helps to  enlarge $\mathpzc{a}$ and reduce $\mathpzc{c}$ that narrows down the parameters' estimation error bound in \eqref{btheta}, whereas maximizing $\frac{\lambda_{min}(S) }{\lambda_{max}(S)}$ matches with the concepts of concurrent learning in continuous-time [2,26]. Note that $P = p+q$ (satisfying $P \ge p+q$) assists to maximize $\frac{ \lambda_{min}(S) }{\lambda_{max}(S)}$ [25] and keeps $\mathpzc{c}$ and $c'$ small. 
\end{remark}

	\section {Simulation Results}
	
	In this section, the performance of the proposed finite-time concurrent learning methods is examined in comparison with asymptotically converging concurrent learning [18] and traditional gradient descent [1] whose estimation laws are, respectively, given as follows,
	\begin{align}
	&{\hat \Theta}(k+1) = {\hat \Theta}(k) - \Gamma_{G}\bar d(k)  e^T(k) ,\nonumber \\
	&{\hat \Theta}(k+1) = {\hat \Theta}(k) - \Gamma_{C}[\Sigma_{G}\bar d(k)  e^T(k) + \Sigma_{C} {  \sum_{h=1}^{P} \bar d(\tau_h)  e_h^T(k)}], \nonumber 
	\end{align}
	where  $\Gamma_{G}=\gamma_{G}I$, $\Gamma_{C}=\gamma_{C}I$, $\Sigma_{G}=\sigma_{G}I$ and $\Sigma_{C}=\sigma_{C}I$ with positive constants $\gamma_{G}>0$, $\gamma_{C}>0$, $\sigma_{G}>0$ and $\sigma_{C}>0$.
	
	The simulation time span is $[k_0,k_f]$ where $k_0=0$ and $k_f=500$, and $\mathcal{D}_x=[x_L,x_H]$ where $x_L<x_H $ and $x_L,x_H \in \Real$ and $\mathcal{D}_x$ is quantized by $[x_L:\frac{x_H-x_L}{k_f-k_0}:x_H]$. In the FTCL methods 1 and 2, given in \eqref{Estrule1} and \eqref{Estrule2}, respectively, $\gamma$ and $\bar\gamma$ are chosen to satisfy \eqref{gamacond} and  \eqref{gamacond1}, respectively, while choosing large values of $\gamma$ and $\bar\gamma$ may jeopardize finite-time convergence. We choose $\xi_G>\xi_C$ and $\bar \xi_G>\bar \xi_C$ to prioritize current data over recorded data and to avoid chattering. Based on [18], for concurrent learning set
	$\gamma_{C} = \frac{1}{2\sigma_{G} +\sigma_{C}\lambda_{max}(S)}$.
	In all cases, the initial values and the controllers are all set to zero. A small exponential sum of sinusoidal input is added to the controller for ensuring the rank condition on the stored data. 
	To fairly compare the precision and speed of all mentioned online learning methods in approximating $f(x)$ and $g(x)$ on the entire domain of $x$, the online learning errors,
	\begin{align}
	E_f(k)= \int _{\mathcal{D}_x} \| e_f(x(k)) \| d^{n} x,\,\,\,  E_g(k)= \int _{\mathcal{D}_x} \| e_g(x(k)) \| d^{n} x,\nonumber
	\end{align}
are calculated online. In this section, the results of the FTCL methods 1 and 2, the traditional concurrent learning and gradient descent methods are respectively labeled by FTCL1, FTCL2, CL and GD.

	\emph{{Example 1: Approximators with zero MFAE ($\bar \varepsilon(k)=0$)}}
	
       Consider the following system 
	  \begin{align}
	  x(k+1)=p_1 e^{-x(k)} + p_2 e^{-x(k)}\cos(x(k)) + p_3 \frac{1}{1+x(k)}u(k),\nonumber
	  \end{align}
	   where the parameters $[p_1,p_2,p_3]$ are unknown and the regressors are fully known as 
	   \begin{align}
	   z(x(k),u(k))=[e^{-x(k)},e^{-x(k)}\cos(x(k)),\frac{1}{1+x(k)}u(k)], \nonumber
	   \end{align}
	    with $p+q=3$. The true unknown parameters are $[p_1,p_2,p_3]=[-1,1.5,1]$ and $\mathcal{D}_x$ is limited with $x_L=0$ and $x_H=2$. We set $P=3$ for both FTCL methods 1 and 2 and concurrent learning method. Let $\gamma_{G}=0.7$ for gradient descent method, $\sigma_{G}=1$, $\sigma_{C}=0.3$ for concurrent learning method, and $\xi_{G}=1$, $\xi_{C}=0.3$ and $\beta=0.3$ for FTCL method 1, $\gamma_{1}=0.6$, $\bar\xi_{G}=1$ and $\bar\xi_{C}=0.3$ for FTCL method 2.
	   Based on the obtained results, the rank condition on $M$ matrix is satisfied in the first $p+q=3$ steps. Therefore, $P$ is chosen as $P=3$ which also satisfies $P\ge p+q$. After $P$ steps, the data selection algorithm [25] is employed to improve the richness of the recorded data for FTCL methods 1 and 2 and concurrent learning method. 
	   
	   Fig. 1 shows the true parameters and the estimated parameters for FTCL methods 1 and 2, concurrent learning and gradient descent methods. In Fig. 1, while gradient descent did not converge to the true parameters, FTCL methods 1 and 2 and concurrent learning succeeded in convergence. However, FTCL methods 1 and 2 resulted in faster convergence to the true parameters in comparison with concurrent learning. The online learning errors $E_f(k)$ and $E_g(k)$ for the FTCL methods 1 and 2, concurrent learning and gradient descent are plotted in Fig. 2 where FTCL methods 1 and 2 show faster convergence to the origin. 
	   The integral absolute errors (IAEs) of $E_f(k)$ and $E_g(k)$ for all methods are compared in Table 1 where FTCL method 2 with IAEs $28.46$ and $54.90$, respectively, for $E_f(k)$ and $E_g(k)$ showed the best precision compared with other methods.  
	\begin{figure}
		\centering
    	\includegraphics[width=3 in,height=2 in]{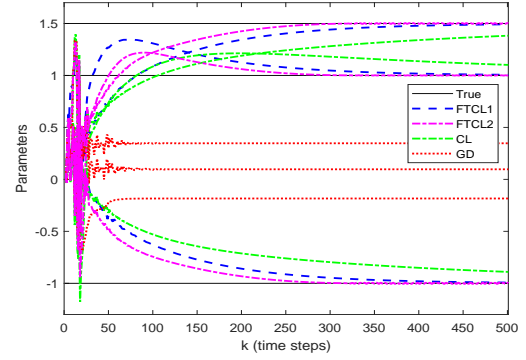}
		\caption{Estimated parameters for Example 1.}
		\label{figurelabel}
	\end{figure}
 \begin{figure}
	\centering
	\includegraphics[width=3 in,height=2 in]{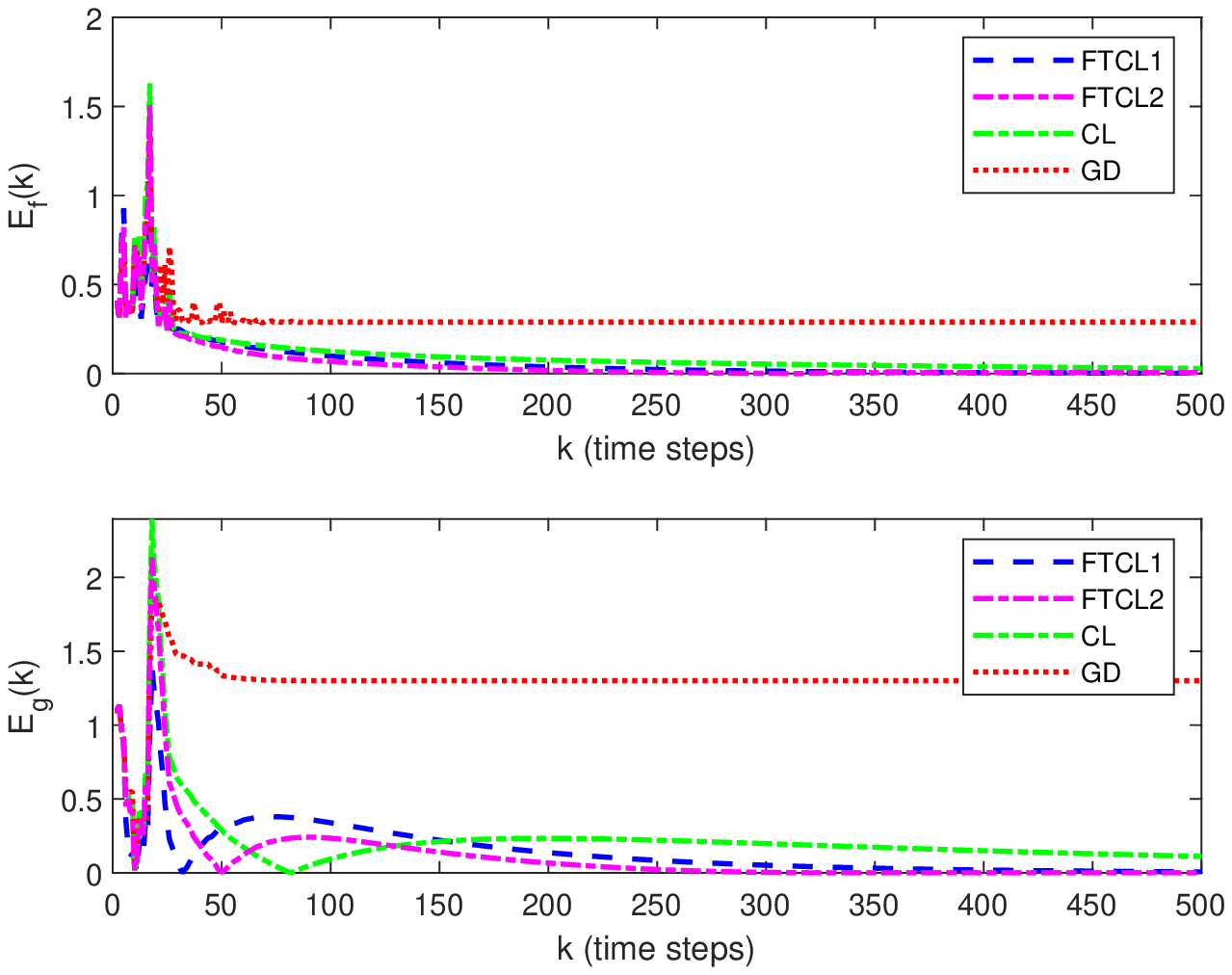}
	\caption{Learning errors for Example 1.}
	\label{figurelabel}
  \end{figure}

\begin{table}
	\begin{center}
		\caption{Learning errors comparison}
		
		\label{tab:table1}
		\centering
		\begin{tabular}{|m{0.8cm} | m{1.5cm} m{1.5cm} | m{1.5cm} m{1.5cm} | } 
			\hline
			{ } & \multicolumn{2}{c|}{Example 1} & \multicolumn{2}{c|} {Example 2}  \\ 
			\text{} & \centering\text{IAE $E_f(k)$} & \centering\text{IAE $E_g(k)$}&\centering\text{IAE $E_f(k)$}& \text{IAE $E_g(k)$} \\
			\hline
			\centering {FTCL2} & \centering{28.46} & \centering{54.90} & \centering{184.39} & {235.23} \\
			\hline
			\centering {FTCL1} & \centering{33.15} & \centering{72.59} & \centering{170.28} & {247.60} \\
			\hline
			\centering{CL} & \centering{51.31} & \centering{113.17} & \centering{234.62} & {269.86} \\ 
			\hline
			\centering{GD} & \centering{152.39} & \centering{645.29} & \centering{635.87} & {675.81} \\ 
			\hline
		\end{tabular}
	\end{center}
\end{table}

\emph{{Example 2: Approximators with non-zero MFAE ($\bar \varepsilon(k) \ne 0$)}}

	Now, consider the following system 
	\begin{align}
	x(k + 1) = 0.5x(k)\sin(0.5x(k)) + (2 + \cos (x(k)))u(k), \label{Ex2}
	\end{align}	
	where the associated $f(x)$ and $g(x)$ are fully unknown uncertainties.
	
	In this example, radial basis function neural networks 
	are employed and 5 radial basis functions $e^ { - \frac{{{{\left\| {x(k) - {c_i}} \right\|}^2}}}{{2{\sigma _i}^2}}}$, $i = 1,2,...,5$
	are used with the centroids $c_i$, uniformly picked on  $\mathcal{D}_x=[x_L,x_H]=[-2,2]$, and the spreads $\sigma _i=1.2$. Hence, the regressor is obtained as follows, 
	\begin{align}
	z(x(k),u(k)) = [&e^{ - \frac{{{{\left\| {x(k) - ( - 2)} \right\|}^2}}}{{2{{(1.2)}^2}}}},...,e^{ - \frac{{{{\left\| {x(k) - (2)} \right\|}^2}}}{{2{{(1.2)}^2}}}},\nonumber\\& e^{ - \frac{{{{\left\| {x(k) - ( - 2)} \right\|}^2}}}{{2{{(1.2)}^2}}}}u(k),...,e^ { - \frac{{{{\left\| {x(k) - (2)} \right\|}^2}}}{{2{{(1.2)}^2}}}}u(k)]^T, \nonumber
	\end{align}
	with 10 independent basis functions ($p+q=10$). The rank condition on $M$ matrix is satisfied in the first $p+q=10$ steps. Thus, $P$ is chosen as $P=10$, satisfying $P \ge p+q$. The data selection algorithm in [25] is employed after the first 10 steps to improve the richness of the recorded data. The approximation of \eqref{Ex2} is given as $x(k + 1) =\hat{\Theta}^T(k) z(x(k),u(k)) =[p_1,p_2,...,p_{10}] z(x(k),u(k))$.
	 Employing $\gamma_{G}=0.8$ for gradient descent method, $\sigma_{G}=1.2$ and $\sigma_{C}=0.1$ for concurrent learning method, $\xi_{G}=1$, $\xi_{C}=0.1$ and $\beta=0.65$ for FTCL method 1,  $\gamma_{1}=0.7$, $\bar\xi_{G}=1$ and $\bar\xi_{C}=0.05$ for FTCL method 2 leads to the estimated parameters $p_1,p_2,...,p_{10}$ depicted on Fig. 3. Due to the lack of persistence of excitation, gradient descent parameters in Fig. 3 could not converge to the appropriate parameters, while FTCL methods 1 and 2 and concurrent learning method 
	 converged to the appropriate parameters. The steady state estimations for $f(x)$ and $g(x)$ are given in Fig. 4. Comparing the learning errors $E_f(k)$ and $E_g(k)$ in Fig. 5 shows that the gradient descent method did not succeed in learning the uncertainty, however FTCL methods 1 and 2 and concurrent learning resulted in bounded learning error convergence near zero. Fig. 5 shows that FTCL methods 1 and 2, in comparison with concurrent learning method are faster in convergence to smaller bounds near zero. 
	 Moreover, based on IAEs of $E_f(k)$ and $E_g(k)$ in Table 1, FTCL methods 1 and 2 result in lower IAEs in comparison with concurrent learning and gradient descent.  
	  \begin{figure}[thpb]
		\centering
		\includegraphics[width=3 in,height=2 in]{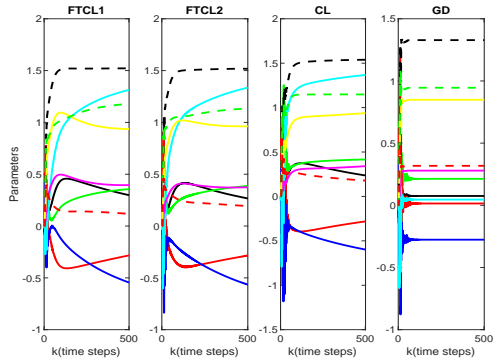}
		\caption{Estimated parameters for Example 2.}
		\label{figurelabel}
	 \end{figure}
 \begin{figure}[thpb]
 	\centering
 	\includegraphics[width=3 in,height=2 in]{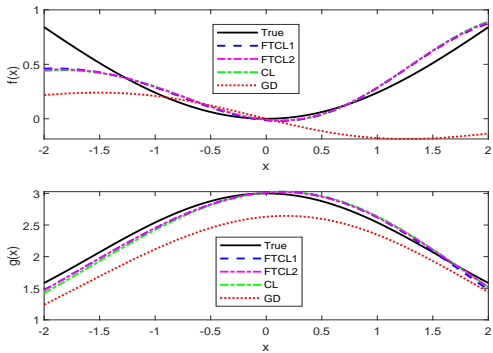}
 	\caption{Steady-state uncertainty approximations.}
 	\label{figurelabel}
 \end{figure}
 \begin{figure}[thpb]
 	\centering
 	\includegraphics[width=3 in,height=2 in]{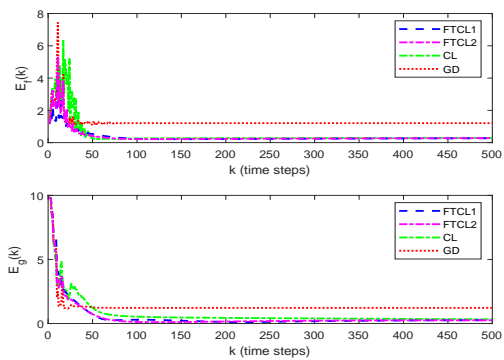}
 	\caption{Learning errors for Example 2.}
 	\label{figurelabel}
 \end{figure}

	\section{Conclusion}
	This paper addressed finite-time identification methods of discrete-time system dynamics where finite-time learning could speed up the learning and concurrent learning technique relaxed the persistence of excitation condition on the regressor to a rank condition on the memory stack of recorded data. For the proposed methods, learning rate conditions were obtained for finite-time convergence based on discrete and finite time analysis. It was discussed that the speed and precision of the presented finite-time identification methods depend on the well-conditioning properties of the memory data. Simulation results are given where it is shown that the presented finite-time concurrent learning methods have better performance in comparison with the traditional gradient descent and asymptotic convergent concurrent learning methods in terms of precision and convergence speed. 

	\section {Appendix: Proof of Theorems 1 and 2}
	\emph{{Proof of Theorem 1.}}
		Consider the Lyapunov function candidate $V(k)$ and its rate of change of $\Delta V(k)$, respectively, as follows
		\begin{align}
		V&(\tilde\Theta(k)) = tr\big\{ \tilde \Theta^T(k) \Gamma^{-1}\tilde \Theta(k)\big\}, \label{VT1}\\
		\Delta V&(\tilde\Theta(k))=V(\tilde\Theta(k))-V(\tilde\Theta(k-1))\nonumber\\&= tr\big\{ \tilde \Theta^T(k) \Gamma^{-1}\tilde \Theta(k)- \tilde \Theta^T(k-1) \Gamma^{-1}\tilde \Theta(k-1)\big\}\nonumber\\& = tr\big\{ (\tilde \Theta(k)-\tilde \Theta(k-1) )^T \Gamma^{-1} (\tilde \Theta(k)+\tilde \Theta(k-1))\big\}.  \label{DelVT1}
		\end{align} 
				
		Using \eqref{errork}, \eqref{errorkhA} and \eqref{Estrule1}, \eqref{DelVT1} can be written as, 
		\begin{align}
		\Delta V&(\tilde\Theta(k))= tr\big\{ \tilde \Theta^T(k-1)[-2\Xi_{G}D(k-1) + \Gamma \Xi_{C}^2 S^2 \nonumber\\& +\Xi_{G}^2 \Gamma  D(k-1)^T D(k-1)+ 2\Gamma\Xi_{G}\Xi_{C}D(k-1)S \nonumber\\& -2\Xi_{C}S] \tilde\Theta(k-1) +2 D^T_{\beta}(k-1) \Gamma \Xi_{C}^2 S \tilde\Theta(k-1)
		\nonumber\\&-2  D^T_{\beta}(k-1)\Xi_{C}  \tilde\Theta(k-1) \nonumber\\&
		 +2  D^T_{\beta}(k-1) \Gamma \Xi_{G} \Xi_{C} D(k-1)  \tilde\Theta(k-1) \nonumber\\&+ 2 [ \bar \varepsilon(k-1) \bar d^T(k-1) \Xi_{G}  -  D^T_{\bar \varepsilon} \Gamma \Xi_{G} \Xi_{C} D(k-1) \nonumber\\&- \bar \varepsilon(k-1) \bar d^T(k-1)  \Xi_{G}^2 \Gamma D(k-1)  -  D^T_{\bar \varepsilon} \Gamma \Xi_{C}^2 S \nonumber\\&  -\bar \varepsilon(k-1) \bar d^T(k-1)  \Gamma \Xi_{G} \Xi_{C} S  +  D^T_{\bar \varepsilon} \Xi_{C}] \tilde \Theta(k-1)   \nonumber\\& +  [\bar \varepsilon(k-1) \bar d^T(k-1)  \Gamma \Xi_{G}^2 -2 D^T_{\beta}(k-1) \Gamma  \Xi_{C} \Xi_{G} \nonumber\\& +2  D^T_{\bar \varepsilon} \Gamma \Xi_{G} \Xi_{C}]  \bar d(k-1) \bar \varepsilon^T(k-1)   + D^T_{\bar \varepsilon} \Gamma \Xi_{C}^2  D_{\bar \varepsilon} \nonumber\\& -2  D^T_{\bar \varepsilon} \Gamma \Xi_{C}^2 D_{\beta}(k-1)  + D^T_{\beta}(k-1) \Gamma \Xi^2_{C} D_{\beta}(k-1)\big\}, \label{DelVT1U2}
		\end{align}
			where $D_{\bar \varepsilon}=\sum_{h=1}^{P}  \bar d(\tau_h) \bar  \varepsilon^T(\tau_h)$, $D(k)=\bar d(k)\bar d^T(k)$ with $ \|D(k)\|<1$ and $D_{\beta}(k)=\frac{  \sum_{h=1}^{P}  \bar d(\tau_h)  e_h^T(k)}{\beta+\|\sum_{h=1}^{P} \bar d(\tau_h)  e_h^T(k)\|} = \frac{  S \tilde \Theta(k) -D_{\bar \varepsilon} }{\beta+\| S \tilde \Theta(k) -D_{\bar \varepsilon} \|}$.
			Based on $\|\bar \varepsilon(k)\| \le b_{\bar \varepsilon}$ and $\|\bar d(k)\| < 1$, we have $\|\bar d(k-1) \bar \varepsilon^T(k-1)\| \le  b_{\bar \varepsilon}$ and $\|D_{\bar \varepsilon} \| \le  {P} b_{\bar \varepsilon}$. Using Fact 1,
		\begin{align}
		\beta+\|S \tilde \Theta(k-1)\| -{P} b_{\bar \varepsilon}  &\le 	\beta+\|S \tilde \Theta(k-1) - D_{\bar \varepsilon} \| \nonumber\\& \le \beta+\|S \tilde \Theta(k-1)\| + {P} b_{\bar \varepsilon},  \label{betaS1}  
		\end{align}
		where $\beta$ is a positive constant satisfying $\beta -P  b_{\bar \varepsilon} > 0$. For some $ \eta > 0 $, one can write 
		\begin{align}
	    \beta+\|S \tilde \Theta(k-1)\| + {P} b_{\bar \varepsilon}  \le (\eta+1)\|S \tilde \Theta(k-1)\|. \label{betaS2}   
		\end{align}
				
		Using \eqref{betaS1} and \eqref{betaS2}, the second, third, fourth and the last terms on the right side of \eqref{DelVT1U2}, are, respectively, upper bounded by 
		\begin{align}
		tr&\big\{2  D^T_{\beta}(k-1) \Gamma \Xi_{C}^2 S \tilde\Theta(k-1)\big\} \nonumber\\& =  tr\big\{  \frac{2\tilde \Theta^T(k-1) S^T \Gamma \Xi_{C}^2 S \tilde \Theta(k-1) - 2  D^T_{\bar \varepsilon} \Gamma \Xi_{C}^2 S \tilde \Theta(k-1)}{\beta+\| S \tilde \Theta(k-1) - D_{\bar \varepsilon} \|} \big\} \nonumber\\& \le  \frac{2 \gamma \xi^2_{C}}{\beta} \lambda_{max}^2 (S) \|\tilde \Theta(k-1)\|^2 +2 \gamma \xi^2_{C} P b_{\bar \varepsilon} \frac{\lambda_{max}(S)}{\lambda_{min} (S)},\label{U1}
		\end{align}
		\begin{align}
		tr&\big\{-2  D^T_{\beta}(k-1) \Xi_{C} \tilde\Theta(k-1) \big\} \nonumber\\ &\le tr\big\{  \frac{ -2\tilde \Theta^T(k-1) S \Xi_{C} \tilde \Theta(k-1) }{\beta+\|S\tilde \Theta(k-1) \| +\|D_{\bar \varepsilon} \|} + \frac{2 D^T_{\bar \varepsilon} \Xi_{C} \tilde \Theta(k-1)}{\beta+\| S \tilde \Theta(k-1) - D_{\bar \varepsilon}\|}\big\}	\nonumber\\& \le tr\big\{  \frac{-2 \tilde \Theta^T(k-1) S \Xi_{C} \tilde \Theta(k-1) }{(\eta +1)\|S\tilde \Theta(k-1) \|}  +  \frac{2 D^T_{\bar \varepsilon} \Xi_{C}\tilde \Theta(k-1)}{\beta+\| S \tilde \Theta(k-1) -D_{\bar \varepsilon} \|} \big\} \nonumber\\  & \le - \frac{2 \xi_{C} }{(\eta +1)} (\frac{\lambda_{min}(S)}{\lambda_{max}(S)}) \|\tilde \Theta(k-1) \| +2 \xi_{C} \frac{{P} b_{\bar \varepsilon}}{\lambda_{min}(S) },\label{U2}
		\\
	  tr&\big\{2  D^T_{\beta}(k-1) \Gamma \Xi_{G} \Xi_{C} D(k-1)  \tilde\Theta(k-1) \big\} \nonumber\\ & \le  tr\big\{2   \frac{  \tilde \Theta^T(k-1) S \Gamma \Xi_{G} \Xi_{C} D(k-1)  \tilde\Theta(k-1) }{\beta} \nonumber\\ &  - 2 \frac{D_{\bar \varepsilon} \Gamma \Xi_{G} \Xi_{C} D(k-1)  \tilde\Theta(k-1)}{\beta+\| S \tilde \Theta(k-1) -D_{\bar \varepsilon} \|}\big\} \nonumber\\& \le   2 \gamma \xi_{G} \xi_{C} \frac{  \lambda_{max}(S) }{\beta} \|\tilde\Theta(k-1)\|^2 + 2 \gamma \xi_{G} \xi_{C}\frac{ {P} b_{\bar \varepsilon} }{\lambda_{min}( S)} ,\label{U3}
	  \\
      tr&\big\{ D^T_{\beta}(k-1) \Gamma \Xi^2_{C}  D_{\beta}(k-1)  \big\} \le  \gamma \xi^2_{C}\|\frac{   S \tilde \Theta(k-1) -D_{\bar \varepsilon}  }{\beta+\| S \tilde \Theta(k-1) -D_{\bar \varepsilon} \|}\|^2  \nonumber\\&\le  \gamma \xi^2_{C}\frac{\|   S \tilde \Theta(k-1) -D_{\bar \varepsilon} \|^2 }{\beta^2}  \le \gamma \xi^2_{C}(\frac{  \|S \tilde \Theta(k-1)\| }{\beta} +  \frac{ \|  D_{\bar \varepsilon} \| }{\beta})^2  \nonumber\\&\le  \frac{\gamma \xi^2_{C}}{\beta^2}({  \lambda_{max}^2(S)  } \|\tilde \Theta(k-1)\|^2 + 2{  \lambda_{max}(S) P b_{\bar \varepsilon}} \|\tilde \Theta(k-1)\| + { P^2 b_{\bar \varepsilon}^2}). \label{U4}
		\end{align}
		Using \eqref{U1}-\eqref{U4} and knowing $\|D_{\beta}(k-1)\| < 1$, it follows that
		\begin{align}
		\Delta V(\tilde \Theta(k)) \le \mathpzc{a} \|\tilde \Theta(k-1)\|^2 + \mathpzc{b} \|\tilde \Theta(k-1)\| + \mathpzc{c}  , \label{DelVab}
		\end{align}
		where $\mathpzc{a}$ and $\mathpzc{c}$ are given in \eqref{a} and \eqref{c}, respectively, and 
		$\mathpzc{b}= - \frac{2 \xi_{C} }{(\eta +1)} (\frac{\lambda_{min}(S)}{\lambda_{max}(S)})+ 2 b_{\bar \varepsilon}(\xi_{C}   \gamma \lambda_{max}(S)(\xi_{G} +  \xi_{C} P(1+ \frac{1}{\beta^2}))  +  \xi_{C} P(\gamma \xi_{G}+ 1)+\xi_{G}  +  \gamma \xi_{G}^2)$.
		One can see from \eqref{VT1} that
		\begin{align}
		V(\tilde \Theta(k)) \le \frac{1}{\gamma}\|\tilde \Theta(k)\|^2  \,\,\, \Rightarrow \,\,\,  \sqrt{\gamma}V^{\frac{1}{2}}(\tilde \Theta(k))\le   \|\tilde \Theta(k)\|. \label{VT1U}
		\end{align}
		Now, we have the following two cases:
		
		\begin{enumerate}
		
		\item For adaptive approximators with zero MFAEs, i.e., $\bar \varepsilon(k)=0$ and $b_{\bar \varepsilon}=0$, using \eqref{VT1U}, \eqref{DelVab} leads to 
		\begin{align}
		\Delta&  V(\tilde \Theta(k))\le -\mathpzc{a_{\gamma}} V(\tilde \Theta(k-1)) -\mathpzc{b_{\gamma}} V^{\frac{1}{2}}(\tilde \Theta(k-1)),\label{DelVagbg}
		\end{align} 
		where $\mathpzc{a_{\gamma}}$ and $\mathpzc{b_{\gamma}}$ are given in \eqref{agama}.
		 Invoking Lemma 1, for $0 < \mathpzc{a_{\gamma}}<1$ and $\mathpzc{b_{\gamma}}>0$ (already satisfied), $\tilde \Theta(k)$ converges to zero within finite-time steps. 
		Satisfying condition \eqref{gamacond} keeps $0 < \mathpzc{a_{\gamma}}$.
	Note that $\mathpzc{a_{\gamma}}<1$ (i.e., $\gamma \mathpzc{a}+1>0$) always hold based on the following inequalities,
	\begin{align}
		 &(1 - {\gamma}\xi_{C}\lambda_{max}(S) )^2   + 2 \gamma^2 \xi_{G} \xi_{C}\lambda_{max}(S) (1+\frac{1}{\beta}) \nonumber\\& + \xi_{G}^2 \gamma^2 + \gamma^2 \xi^2_{C} \lambda_{max}^2(S)  (\frac{2}{\beta}+\frac{1}{\beta^2}) > 0\Rightarrow \nonumber\\ &-2{\gamma}\xi_{C}\lambda_{min}(S) + \gamma^2 \xi_{C}^2 \lambda_{max}^2(S) + \gamma^2 \xi^2_{C} \lambda_{max}^2(S)  (\frac{2}{\beta}+\frac{1}{\beta^2})   \nonumber\\& + 2 \gamma^2 \xi_{G} \xi_{C}\lambda_{max}(S) (1+\frac{1}{\beta})
		  + \xi_{G}^2 \gamma^2 +1 > 0 .
		\end{align}
	Therefore, based on Lemma 1 by satisfying \eqref{gamacond}, $\tilde \Theta(k)$ converges to zero for $k \ge K_1^*(V(\tilde \Theta_0))$.
	Using Lemma 1, one obtains \eqref{k1} for the associated settling-time function.
	
	\item For adaptive approximators with non-zero MFAEs, i.e., $\bar \varepsilon(k) \ne 0$, it is known that $\mathpzc{c}>0$ and $\mathpzc{a}<0$ (satisfying \eqref{gamacond}). Thus, by bounding $\mathpzc{b}$ with $\mathpzc{b_u}$, given in \eqref{bu}, from \eqref{DelVab} one obtains
		\begin{align}
		\Delta V(\tilde \Theta(k)) \le \mathpzc{a} \|\tilde \Theta(k-1)\|^2 +\mathpzc{b_u} \|\tilde \Theta(k-1)\| + \mathpzc{c}. \label{DelVabu}
		\end{align}
		Since, $\|\tilde \Theta(k)\| \ge 0$ and $\mathpzc{a}<0$, 
		the only valid non-negative root of \eqref{DelVabu} is $b_{\tilde\Theta} = \frac{-\mathpzc{b_u} - \sqrt{(\mathpzc{b_u})^2 - 4(\mathpzc{a})\mathpzc{c}}}{2\mathpzc{a}}$.
		Thus, if $\|\tilde \Theta(k)\| > b_{\tilde\Theta}$, then $\Delta V(\tilde \Theta(k))<0$,
		whereas, after $\tilde \Theta(k)$ enters the set $
		S_{\tilde \Theta} = \big\{\tilde \Theta : \|\tilde \Theta\| \le b_{\tilde\Theta} \big\}$, it is possible to have $\Delta V(k) \ge 0$. However, for discrete-time samples thereafter, $\tilde \Theta(k)$ stay within the positive invariant set $S_{\tilde \Theta} $. Therefore, provided that $\|\tilde \Theta_0\| > b_{\tilde\Theta}$, for all $k$, $\|\tilde \Theta(k)\| \le b_{\tilde\Theta}$.
		Hence, 
		$\tilde \Theta(k)$ is finite-time attractive to $S_{\tilde \Theta}, $ being invariant. 
				To obtain the settling-time function $K_2^*(V(\tilde \Theta_0))$ that $\|\tilde \Theta(k)\| $ reaches the invariant set $S_{\tilde \Theta} $, 
				using \eqref{VT1U}, \eqref{DelVabu} is written as 
		 \begin{align}
		 \Delta V(\tilde \Theta(k)) \le -\mathpzc{a_{\gamma}} V(\tilde \Theta(k-1)) +\mathpzc{b_u} \|\tilde \Theta(k-1)\| + \mathpzc{c}. \label{DelVagbu}
		 \end{align}
		 
		Then, using \eqref{DelVagbu} for $k=1,...,K_2^*-1,K_2^*$, one has
		 \begin{align}
		&V(\tilde \Theta(1))-V(\tilde \Theta_0) \le -\mathpzc{a_{\gamma}} V(\tilde \Theta_0) +\mathpzc{b_u} \|\tilde \Theta_0\| + \mathpzc{c},\nonumber \\
		&V(\tilde \Theta(2))-V(\tilde \Theta(1)) \le -\mathpzc{a_{\gamma}} V(\tilde \Theta(1)) +\mathpzc{b_u} \|\tilde \Theta_0\| + \mathpzc{c},\nonumber \\
		&\vdots \nonumber \\
		&V(\tilde \Theta(K_2^*))-V(\tilde \Theta(K_2^*-1)) \le -\mathpzc{a_{\gamma}} V(\tilde \Theta(K_2^*-1))\nonumber \\&+\mathpzc{b_u} \|\tilde \Theta_0\| + \mathpzc{c},\nonumber
		\end{align}
		with $V(\tilde \Theta(k))<V(\tilde \Theta(k-1))$ which results in
		\begin{align}
		 V(\tilde \Theta(K_2^*))-V(\tilde \Theta_0) \le K_2^*(-\mathpzc{a_{\gamma}} V(\tilde \Theta(K_2^*)) +\mathpzc{b_u} \|\tilde \Theta_0\|  + \mathpzc{c}).  \label{DelVk2}
		 \end{align}
		 Using $ V(\tilde \Theta(K_2^*)) \le \gamma^{-1} (b_{\tilde\Theta})^2$, the above inequality leads to the settling-time function given in \eqref{k2}.
		 Therefore, for $k \ge K_2^*(V(\tilde \Theta_0))$, $\tilde \Theta(k)$ reaches the invariant set $S_{\tilde \Theta}$.
		This completes the proof.\frQED
\end{enumerate}	
	
\emph{Proof of Theorem 2.}	Consider $V(\tilde \Theta(k))$ and $\Delta V(\tilde \Theta(k))$, respectively, given in \eqref{VT1} and \eqref{DelVT1}. Using \eqref{Estrule2}, \eqref{DelVT1} is written as, 
			\begin{align}
			&\Delta V(\tilde \Theta(k))=tr\big\{-2 \bar \Xi_{G}  \lfloor e(k-1) \rceil ^ {\gamma_1}  \bar d^T(k-1) \tilde \Theta(k-1)    \nonumber \\&
			- 2\bar \Xi_{C} \sum_{h=1}^{P}  \lfloor e_h(k-1)\rceil ^ {\gamma_1}   \bar d^T(\tau_h) \tilde \Theta(k-1)  \nonumber \\&
			+\bar \Gamma \bar \Xi_{G}^2  \lfloor e(k-1) \rceil ^ {\gamma_1}   D(k-1) \lfloor e^T(k-1)\rceil ^ {\gamma_1}  \nonumber \\& 
			+\bar \Gamma \bar \Xi_{C} \bar \Xi_{G} \sum_{h=1}^{P}  \lfloor e_h(k-1) \rceil ^ {\gamma_1}  \bar d^T(\tau_h) \bar d(k-1) \lfloor e^T(k-1) \rceil ^ {\gamma_1}  \nonumber \\& 
			+\bar \Gamma \bar \Xi_{C} \bar \Xi_{G} \lfloor e(k-1)\rceil ^ {\gamma_1}  \bar d^T(k-1)(\sum_{h=1}^{P} \bar d(\tau_h) \lfloor e_h^T(k-1) \rceil ^ {\gamma_1} )
			\nonumber \\&
			+\bar \Gamma \bar \Xi_{C}^2 \sum_{h=1}^{P}\lfloor e_h(k-1)\rceil ^ {\gamma_1}  \bar d^T(\tau_h)\sum_{h=1}^{P}\bar d(\tau_h)\lfloor e_h^T(k-1)\rceil ^ {\gamma_1} \big\}.  \label{DelV6}  
		\end{align} 
	Consider in the component-wise sense that $|(\bar d^T(k) \tilde \Theta(k))_j| \ge |(\bar \varepsilon(k))_j|$, for $j=1,...,n$. 
	Therefore, $sign(\bar d^T(k) \tilde \Theta(k) - \bar \varepsilon^T(k))= sign(\bar d^T(k) \tilde \Theta(k))$. Then, for any $y, \bar y \in \Real$ and $ 0<\gamma_1<1 $, one has $|y + \bar y |^{\gamma_1} <|y|^{\gamma_1} +|\bar y|^{\gamma_1}$ [22].
		Thus, defining $y=(\bar d^T(k) \tilde \Theta(k))_j- (\bar \varepsilon(k))_j$ and $\bar y = (\bar \varepsilon(k))_j$, for all $j=1,...,n$, one obtains that $|(\bar d^T(k) \tilde \Theta(k))_j |^{\gamma_1} -|( \bar \varepsilon(k))_j|^{\gamma_1}  \le |(\bar d^T(k) \tilde \Theta(k))_j - (\bar \varepsilon(k))_j|^{\gamma_1}$,
		and then in the component-wise sense,  
				\begin{align}
			-|\bar d^T(k) \tilde \Theta(k) -\bar \varepsilon(k) |^{\gamma_1}  \le -|\bar d^T(k) \tilde \Theta(k) |^{\gamma_1} + | \bar \varepsilon(k)|^{\gamma_1}. \label{ineq20}
		\end{align}
One knows that 
\begin{align}
(\lfloor \bar d^T(k) \tilde \Theta(k) \rceil ^ {\gamma_1})^T  \bar d^T(k) \tilde \Theta^T(k)  &= \| \bar d^T(k)\tilde \Theta(k)\|^{\gamma_1+1}_{\gamma_1+1}, \label{F1} \\
 \| \lfloor \bar d^T(\tau_h)\tilde \Theta(k-1) \rceil ^{\gamma_1}\| &= \| \bar d^T(\tau_h)\tilde \Theta(k-1)  \|^{\gamma_1} _{2\gamma_1}.\label{F2} 
	\end{align}
	Note that $\||\bar \varepsilon(k)|^{\gamma_1}\| =  \|\bar \varepsilon(k)\|^{\gamma_1}_{2\gamma_1}$ and by Fact 2 one has,
		\begin{align}
			\|\bar \varepsilon(k)\|_{2 \gamma_1} &\le n^{\frac{1- \gamma_1}{2 \gamma_1}} \|\bar \varepsilon(k)\|, \label{epsi}
			\\	\|\bar d^T(k) \tilde \Theta(k)\| &\le \|\bar d^T(k) \tilde \Theta(k)\|_{\gamma_1+1}, \label{F3} 
		\end{align}
for all $0<2\gamma_1<2$. 
	Now, using \eqref{errork}, \eqref{errorkhA} and \eqref{ineq20}-\eqref{F3}, one obtains
		\begin{align}
			&\Delta V(\tilde \Theta(k))\le \nonumber \\& -2 \bar \xi_{G} ( \| \bar d^T(k-1) \tilde \Theta(k-1)  \| ^ {\gamma_1+1}_ {\gamma_1+1}  -\| \tilde \Theta(k-1)  \| \||\bar \varepsilon(k-1)|^{\gamma_1}\|) \nonumber \\&
			- 2\bar \xi_{C} \sum_{h=1}^{P}  (\| \bar d^T(\tau_h) \tilde \Theta(k-1)\| ^ {\gamma_1+1} _ {\gamma_1+1} -  \|\bar d^T(\tau_h) \tilde \Theta(k-1)  \| \||\bar \varepsilon(\tau_h)|^{\gamma_1}\| )\nonumber \\&
			+\bar\gamma \bar \xi_{G}^2 \| \bar d^T(k-1) \tilde \Theta(k-1) \| ^ {2\gamma_1}_{2\gamma_1} +\bar\gamma \bar \xi_{G}^2 \||\bar \varepsilon(k-1)|^{\gamma_1}\|^2 \nonumber \\&+  2\bar \gamma \bar \xi_{G}^2 \| \bar d^T(k-1) \tilde \Theta(k-1) \| ^ {\gamma_1}_{2\gamma_1} \||\bar \varepsilon(k-1)|^{\gamma_1}\|\nonumber \\& 
			+2\bar\gamma \bar \xi_{C} \bar \xi_{G} \sum_{h=1}^P  \| \bar d^T(\tau_h) \tilde \Theta(k-1) \|^ {\gamma_1}_ {2\gamma_1}  \| \bar d^T(k-1) \tilde \Theta(k-1) \|^ {\gamma_1}_ {2\gamma_1} \nonumber \\& +2\bar\gamma \bar \xi_{C} \bar \xi_{G} {P} \||\bar\varepsilon(\tau_h)|^{\gamma_1}\| \||\bar\varepsilon(k-1)|^{\gamma_1}\| \nonumber \\& +2\bar\gamma \bar \xi_{C} \bar \xi_{G} \sum_{h=1}^P \| \bar d^T(\tau_h) \tilde \Theta(k-1) \|^ {\gamma_1}_ {2\gamma_1} \||\bar\varepsilon(k-1)|^{\gamma_1}\|
			\nonumber \\& +2\bar \gamma \bar \xi_{C} \bar \xi_{G} {P} \| \bar d^T(k-1) \tilde \Theta(k-1) \|^ { \gamma_1}_ {2 \gamma_1} \||\bar\varepsilon(\tau_h)|^{\gamma_1}\|
			\nonumber \\&
			 + \bar \gamma \bar \xi_{C}^2 \sum_{h=1}^P\|\bar d^T(\tau_h) \tilde \Theta(k-1)  \|^{2\gamma_1}_{2\gamma_1} + \bar \gamma \bar \xi_{C}^2 {P^2}\||\bar \varepsilon (\tau_h) |^{ \gamma_1}\|.   \label{DelV8}
		\end{align}
Using \eqref{epsi}, $\|\bar \varepsilon(k)\| \le b_{\bar \varepsilon}$ and Fact 2, 
\eqref{DelV8} leads to
		\begin{align}
			&\Delta V(\tilde \Theta(k))\le  (-2 \bar \xi_{G} +\bar \gamma \bar \xi_{G}^2 n^{1-\gamma_1}) \| \bar d^T(k-1) \tilde \Theta(k-1)  \| ^ { \gamma_1+1}\nonumber \\& + (- 2\bar \xi_{C} +n^{1- \gamma_1}(2\bar \gamma \bar \xi_{C} \bar \xi_{G}  +\bar \gamma \bar \xi_{C}^2  )) \sum_{h=1}^{P}  \| \bar d^T(\tau_h) \tilde \Theta(k-1)\| ^ { \gamma_1+1} \nonumber \\& + 2 b_{\bar \varepsilon}^{\gamma_1}(\bar \xi_{C} n^{\frac{1- \gamma_1}{2}} +n^{1-\gamma_1}(\bar \gamma \bar\xi_{C} \bar \xi_{G}  +\bar \gamma \bar \xi_{C}^2 )) \sum_{h=1}^{P}  \|\bar d^T(\tau_h) \tilde \Theta(k-1)  \|\nonumber \\& + 2 b_{\bar \varepsilon}^{ \gamma_1} n^{1- \gamma_1}\bar \gamma ( \bar\xi_{G}^2  +\bar\xi_{C} \bar \xi_{G}  P ) \| \bar d^T(k-1) \tilde \Theta(k-1) \|  +  2 \bar \xi_{G} n^{\frac{1- \gamma_1}{2}} \times \nonumber \\ & b_{\bar \varepsilon}^{ \gamma_1}\| \tilde \Theta(k-1)  \|   +\bar \gamma n^{1-{ \gamma_1}} b_{\bar \varepsilon}^{2 \gamma_1}( \bar\xi_{G}^2 +2 \bar \xi_{C} \bar \xi_{G} {P}+  \bar \xi_{C}^2 {P^2} n^{\frac{\gamma_1-1}{2}} b_{\bar \varepsilon}^{ -\gamma_1}).  \label{DelV19}
		\end{align}
Using $S=\sum_{h=1}^{P} \bar d(\tau_h)\bar d^T(\tau_h)$, \eqref{DelV19} is rewritten as
     	\begin{align}
			\Delta V&(\tilde \Theta(k))\le -a' \| \tilde \Theta(k-1)\|^ {\gamma_1+1} + b' \| \tilde \Theta(k-1)\| +  c' ,  \label{DelV20}
		\end{align}
		where $	a'$, $b'$, $c'$ are given in \eqref{aprime}-\eqref{cprime}. 
		One can see that using \eqref{VT1U}, \eqref{DelV20} can be written as
\begin{align}
			\Delta V&(\tilde \Theta(k))\le -\alpha' V^ {\frac{\gamma_1+1}{2}}(\tilde \Theta(k-1)) + b' \| \tilde \Theta(k-1)\| +  c'   ,  \label{DelV15}
		\end{align} where $\alpha'=a' \bar\gamma^{\frac{\gamma_1+1}{2}}$. Now, we have the following two cases:
		\begin{enumerate}
		\item 	For adaptive approximators with zero MFAEs ($\bar \varepsilon(k)=0$) and $b_{\bar \varepsilon}=0$, \eqref{DelV15} reduces to $\Delta V(\tilde \Theta(k))\le -\alpha' V^ {\frac{\gamma_1+1}{2}}(\tilde \Theta(k-1))$, with $0<\frac{\gamma_1+1}{2}<1$
and this leads to 
		\begin{align}
			\Delta V&(\tilde \Theta(k))  \le -\alpha' \min  \{ \frac{V(\tilde \Theta(k-1))}{\alpha'}, V^ {\frac{\gamma_1+1}{2}}(\tilde \Theta(k-1))\} .   \label{DelV16}
		\end{align}
In order to have $0<\alpha'$, the condition \eqref{gamacond1} should be satisfied.
Invoking Lemma 2 and \eqref{DelV16}, $\tilde \Theta(k)$ converges to zero and a settling-time function is obtained as given in \eqref{K111}.   
\item For adaptive Approximators with non-zero MFAEs ($\bar \varepsilon(k) \ne 0$), by considering $\gamma_1$ as $\gamma_1=\frac{m_1}{m_2}$, one finds the roots of \eqref{DelV20} where $m_1$ and $m_2$ are positive integers with $m_1<m_2$. Using Descartes' rule of signs, \eqref{DelV20} contains a positive root $\bar b_{\tilde \Theta}$ where for $\|	\tilde \Theta(k) \| > \bar b_{\tilde \Theta}$, one has $\Delta V(\tilde \Theta(k)) <0$.  Thus, similar to the proof of part 2 in Theorem 1, one shows that $\tilde \Theta(k)$ is finite-time attractive to the bound $\bar b_{\tilde \Theta}$ (i.e., $\|{\tilde \Theta}(k)\| \le \bar b_{\tilde \Theta}$) and one obtains the corresponding settling-time function as given in \eqref{k22}. This completes the proof.\frQED
\end{enumerate}
\end{document}